\documentclass[12pt, draftclsnofoot, onecolumn]{IEEEtran}
\usepackage{amsmath}
\usepackage{amssymb}
\usepackage{graphicx}
\usepackage{epsfig}
\usepackage{subfigure}
\usepackage{mathrsfs}
\usepackage{fancyhdr}

\newtheorem{theorem}{Theorem}

\newcommand{\PP}{\mathbb{P}}
\newcommand{\E}{\mathbb{E}}

\newcommand{\sinr}{{\rm SINR}}

\newcommand{\dd}{{\rm d}}
\newcommand{\ie}{{\em i.e., }}

\newcommand{\TR}{T_{\mathrm{FR}}}
\newcommand{\TF}{T_{\mathrm{FFR}}}

\begin{document}

\title{Analytical Evaluation of Fractional Frequency Reuse for OFDMA Cellular Networks}

\author{\IEEEauthorblockN{Thomas David Novlan, Radha Krishna Ganti, Arunabha Ghosh, Jeffrey G. Andrews\\}
}

\maketitle
\pagestyle{fancy}
\thispagestyle{empty}
\fancyhf{} 
\fancyhead[R]{\thepage}
\renewcommand{\headrulewidth}{0pt}
\renewcommand{\footrulewidth}{0pt}

\begin{abstract}
Fractional frequency reuse (FFR) is an interference management technique well-suited to OFDMA-based cellular networks wherein the cells are partitioned into spatial regions with different frequency reuse factors. To date, FFR techniques have been typically been evaluated through system-level simulations using a hexagonal grid for the base station locations. This paper instead focuses on analytically evaluating the two main types of FFR deployments - Strict FFR and Soft Frequency Reuse (SFR) - using a Poisson point process to model the base station locations. The results are compared with the standard grid model and an actual urban deployment. Under reasonable special cases for modern cellular networks, our results reduce to simple closed-form expressions, which provide insight into system design guidelines and the relative merits of Strict FFR, SFR, universal reuse, and fixed frequency reuse. We observe that FFR provides an increase in the sum-rate as well as the well-known benefit of improved coverage for cell-edge users. Finally, a SINR-proportional resource allocation strategy is proposed based on the analytical expressions, showing that Strict FFR provides greater overall network throughput at low traffic loads, while SFR better balances the requirements of interference reduction and resource efficiency when the traffic load is high. 
\end{abstract}

\let\thefootnote\relax\footnotetext{T. D. Novlan, R. K. Ganti, and J. G. Andrews are with the Wireless Networking and Communications Group, the University of Texas at Austin. A. Ghosh is with AT\&T Laboratories. The contact author is J. G. Andrews. Email: jandrews@ece.utexas.edu. This research has been supported by AT\&T Laboratories. Date revised: January 23, 2011} 

\linespread{2}
\section{\label{sec:intro} Introduction}
Modern multi-cellular systems feature increasingly dense base station deployments in an effort to provide higher network capacity as user traffic, especially data traffic, increases. Because of the soon ubiquitous use of Orthogonal Frequency Division Multiple Access (OFDMA) in these networks, the intra-cell users are assumed to be orthogonal to each other and the primary source of interference is inter-cell interference, which is especially limiting for users near the boundary of the cells. Inter-cell interference coordination (ICIC) is a strategy to improve the performance of the network by having each cell allocate its resources such that interference experienced in the network is minimized, while maximizing spatial reuse. 

Fractional frequency reuse (FFR) has been proposed as an ICIC technique in OFDMA based wireless networks \cite{Boudreau2009}. The basic idea of FFR is to partition the cell's bandwidth so that (i) cell-edge users of adjacent cells do not interfere with each other and (ii) interference received by (and created by) cell-interior users is reduced, while (iii) using more total spectrum than conventional frequency reuse. The use of FFR in cellular networks leads to natural tradeoffs between improvement in rate and coverage for cell-edge users and sum network throughput and spectral efficiency. Most prior work resorted to simulations to evaluate the performance of FFR, primarily because of the intractability of the hexagonal grid model of base station locations. In this paper, instead, we model the BS locations as a Poisson point process (PPP).  One advantage of this approach is the ability to capture the non-uniform layout of modern cellular deployments due to topographic, demographic, or economic reasons \cite{Baccelli1995}, \cite{Brown2000}, \cite{Haenggi2009}. Additionally, tractable expressions can be drawn from the Poisson model, leading to more general performance characterizations and intuition \cite{Andrews2010}. 

\subsection{Fractional Frequency Reuse}
There are two common FFR deployment modes: \textit{Strict FFR} and \textit{Soft Frequency Reuse} (SFR). While FFR can be considered in the uplink or downlink, this work focuses on the downlink since it typically supports links with greater rate requirements with a low margin for interference and additionally we can, unlike the uplink, neglect power control by assuming equal power downlinks. 

\subsubsection{Strict FFR} Strict FFR is a modification of the traditional frequency reuse used extensively in multi-cell networks \cite{Begain2002,Sternad2003}. Fig. \ref{fig:FFRmodels}(a) illustrates Strict FFR for a hexagonal grid modeled deployment with a cell-edge reuse factor of $\Delta=3$. Users in each cell-interior are allocated a common sub-band of frequencies while cell-edge users' bandwidth is partitioned across cells based on a reuse factor of $\Delta$. In total, Strict FFR thus requires a total of $\Delta+1$ sub-bands. Interior users do not share any spectrum with exterior users, which reduces interference for both interior users and cell-edge users. 

\subsubsection{SFR} Fig. \ref{fig:FFRmodels}(b) illustrates a SFR deployment with a reuse factor of $\Delta=3$ on the cell-edge. SFR employs the same cell-edge bandwidth partitioning strategy as Strict FFR, but the interior users are allowed to share sub-bands with edge users in other cells. Because cell-interior users share the bandwidth with neighboring cells, they typically transmit at lower power levels than the cell-edge users \cite{Li1999,Huawei2005}. While SFR is more bandwidth efficient than Strict FFR, it results in more interference to both cell-interior and edge users \cite{Doppler2009}. 

\subsection{Related Work and Contributions}
Recent research on FFR has focused on the optimal design of FFR systems by utilizing advanced techniques such as graph theory \cite{Chang2009} and convex optimization \cite{Assad2008} to maximize network throughput. Additional work considers scheduling \cite{Fang2008,Ali2009,Rahman2010} and the authors determine the frequency partitions in a two-stage heuristic approach. These along with other related works utilize the standard equally-spaced grid model for the base stations which do not result in closed or intuitive expressions for $\sinr$, probability of coverage (or outage), or rate, and numerical simulations are used to validate the proposed model or algorithm \cite{Fujii2008,Novlan2010,Alsawah2008,Chen2010}.

The primary contribution of this work is a new analytical framework to evaluate coverage probability and average rate in Strict FFR and SFR systems. These are important metrics to consider, especially for users at the cell-edge since modern cellular networks are increasingly required to provide users with high data-rate and guaranteed quality-of-service, regardless of their geographic location, instead of simply a minimum $\sinr$ which may be acceptable for applications like voice traffic. Through a comparison with an actual urban base station deployment, we show that the grid model provides an upper bound for actual performance since it idealizes real network geometry, while our framework, based on the Poisson model, is a lower bound.

In addition, by considering a special case relevant to interference-limited networks, the analytical expressions for the $\sinr$ distributions reduce to simple expressions which are a function of the key FFR design parameters. We use this analysis to develop system guidelines which show that while Strict FFR provides better coverage probability for edge users than SFR for low power control factors, a SFR system can improve its coverage performance by increasing the cell-edge user power control factor, approaching the performance of per-cell frequency reuse, without the loss in spectral efficiency that is inherent in Strict FFR. Finally, this work presents a strategy for optimally allocating frequency sub-bands to edge users for SFR and FFR based on a chosen threshold $\TR$, which can be related to network traffic load. Numerical results show that the $\sinr$-proportional resource allocation strategy gives insight for choosing FFR parameters that maximize sum rate over universal or per-cell reuse while efficiently allocating sub-bands to provide increased coverage to edge users for given traffic load or coverage requirements. In the next section, we provide a detailed description of the system model and our assumptions.
 
\section{\label{sec:model}System Model}
We consider an OFDMA cellular downlink. We assume that the mobile user is served by its closest base station. The base station locations are distributed 
as a spatial Poisson point process (PPP). We assume that all the BSs transmit with an equal power $P$. The path loss exponent is given by $\alpha$, and ${\sigma}^2$ is the noise power. We assume that the small-scale fading between any BS $z$ and and the typical mobile in consideration, denoted by $g_z$, is i.i.d exponentially distributed with mean $\mu$ (corresponds to Rayleigh fading). The set of interfering base stations is $\mathcal{Z}$, i.e. base stations that use the same sub-band as user $y$. We denote the distance between the interfering BS $z$ in and the mobile node in consideration $y$ by $R_z$. 

The associated Signal to Interference Plus Noise Ratio ($\sinr$) is given as
\begin{equation}
\label{eq:SINR}
\sinr = \frac{Pg_yr^{-\alpha}}{{\sigma}^2 + PI_r},
\end{equation}
where for an interfering BS set $\mathcal{Z}$,
\begin{equation}
\label{eq:FFRIR}
I_r = \displaystyle\sum_{z\in \mathcal{Z}}{g_z{R_z}^{-\alpha}}. 
\end{equation}
In the above expression, we have assumed that the nearest BS to the mobile $y$ is at a distance $r$, which is a random variable.

Additionally, Strict FFR and SFR classify two types of users: \textit{edge} and \textit{interior} users. These classifications come from the typical grid model assumption for the base stations in which constant $\sinr$ contours can be defined as concentric circles around the central base station \cite{Hernandez2009}. In this work however, since the BS locations are distributed as a PPP, the term edge or interior user does not have the same geographic interpretation. Each cell is a Voronoi region with a random area \cite{Haenggi2009} which, as noted in \cite{Brown2000}, more closely reflects actual deployments which are highly non-regular and provides a lower bound on performance metrics due to the lack of repulsion between base stations, which may be arbitrarily close together. Instead, a more general case is considered, in which a base station classifies users with average $\sinr$ less than a pre-determined threshold $\TR$ as edge users, while users with average $\sinr$ greater than the threshold are classified as interior users. Thus the FFR threshold $\TR$ is a design parameter analogous to the grid-based interior radius \cite{Novlan2010}.  

In the case of SFR, inter-cell interference $I_r$ no longer comes from disjoint sets of interior and edge downlinks, but can come from either set and coarse power control is typical \cite{Huawei2005}. To accomplish this, a power control factor $\beta \geq 1$ is introduced to the transmit power to create two different classes, $P_{\rm int} = P$ and $P_{\rm edge} = \beta P$, where $P_{\rm int}$ is the transmit power of the base station if user $y$ is an interior user and $P_{\rm edge}$ is the transmit power of the base station if user $y$ is a cell-edge user.

The interfering base stations are also separated into two classes: $\mathcal{I}_{\rm int}$, which consists of all interfering base stations transmitting to cell-interior users on the same sub-band as user $y$ (at power $P_{\rm int}$) and $\mathcal{I}_{\rm edge}$, which consists of all interfering base stations transmitting to cell-edge users on the same sub-band as user $y$ (at power $P_{\rm edge}$). For a cell-edge user $y$, the resulting out-of-cell interference expression $I_r$ with SFR is given as
\begin{equation}
\label{eq:SFRIR}
I_r = \displaystyle\sum_{i\in \mathcal{I}_{\rm int}}{g_{i}{R_{i}}^{-\alpha}} + \beta \displaystyle\sum_{i\in \mathcal{I}_{\rm edge}}{g_i}{R_{i}}^{-\alpha} 
\end{equation}
Typical analysis of SFR uses values of 2-20 for $\beta$, although this choice is usually based on heuristic results \cite{Doppler2009,Mazin2010}.

\section{\label{sec:coverage} Coverage Probability}
This section presents general coverage probability expressions and numerical results for the two FFR systems. Coverage is an important metric to consider since it can have a large impact on cell-edge user QoS and when combined with resource efficiency, can give an overall picture of cell/network capacity. In the context of this paper, we define coverage probability $p_c$ as the probability that a user's instantaneous $\sinr$ is greater than a value $T$:
\begin{equation}
\label{eq:coverage}
p_c = \PP(\sinr > T)   
\end{equation}
This coverage probability $p_c$ is equivalently the CCDF of the $\sinr$ for a particular reuse strategy, which we will denote as $\bar{F}(T)$. 

In the case of past work, using the grid model, base stations are assumed to be on a hexagonal or rectangular grid, allowing these expressions to be numerically computed. Also, approximations using the symmetric structure of the far-out tiers in the deployment may be employed, or a worst-case user location at the edge of the cell may be considered \cite{Elayoubi2007}. However, the results of this section take advantage of the framework recently developed in \cite{Andrews2010}. The base station locations are instead modeled as a Poisson point process (PPP). Despite the new source of randomness in the model, the authors of \cite{Andrews2010} in Theorem 4 give a general expression for the coverage probability of a typical mobile as a function of the $\sinr$ threshold $T$ for a given base station density $\lambda$, pathloss factor $\alpha$ and $\Delta$ number of frequency sub-bands as 
\begin{equation}
p_c(T,\lambda,\alpha,\Delta) = \pi \lambda \int_0^\infty  e^{-\pi \lambda v (1 + \frac{1}{\Delta}\rho(T,\alpha)) -\mu T \frac{\sigma^2}{P} v^{\alpha/2}} \dd v,
\label{eq:pc}
\end{equation}
where
\begin{equation}
\rho(T,\alpha)= T^{2/\alpha}\int_{T^{-2/\alpha}}^\infty \frac{1}{1+u^{\alpha/2}} \dd u.
\end{equation}
We now provide the distribution of $\sinr$ for Strict FFR and SFR.

\subsection{Strict FFR} The first result using the Poisson model focuses on the $\sinr$ distribution of cell edge users. In the case of Strict FFR, these are the users who have $\sinr$ less than the reuse threshold $\TR$ on the common sub-band shared by all cells and are therefore selected by the reuse strategy to have a new sub-band allocated to them from the $\Delta$ total available sub-bands reserved for the edge users. 

\begin{theorem}[Strict FFR, edge user]
\label{thm:FFRe}
The coverage probability of an edge user in a strict FFR system, assigned a FFR sub-band is 
\[\bar{F}_{\mathrm{FFR,e}}(T) =\frac{p_{c}(T,\lambda,\alpha,\Delta) - \pi \lambda \int_0^\infty  e^{-\pi \lambda v (1 + 2\xi(T,\TR,\alpha,\Delta)) -\mu (T+\TR) \frac{\sigma^2}{P} v^{\alpha/2}} \dd v}{1-p_c(\TR,\lambda,\alpha,1)},\]
where 
\[\xi(T,\TR,\alpha,\Delta)= \int_1^\infty\left[1- \frac{1}{1+\TR x^{-\alpha}}\left(1-\frac{1}{\Delta}\left(1-\frac{1}{1+T x^{-\alpha}}\right) \right) \right] x \dd x,\]
and $p_c(T,\lambda,\alpha,\Delta)$ is given by (\ref{eq:pc}).
\end{theorem}

\begin{proof}
The proof is given in Appendix A.
\end{proof}

While $\xi(T,\TR,\alpha,\Delta)$ is reminiscent of $\rho(T,\alpha)$ given by prior results in \cite{Andrews2010}, it differs due to the dependence of the user's $\sinr$ before and after the assignment of the new FFR sub-band. This is from the fact that while the interference power and the user's fading values have changed, the location of the user relative to the base stations has not changed, and the thus the dominant path loss remains the same.

Now we turn our attention to the interior users in the case of Strict FFR. The coverage probability of the inner user does not depend on $\Delta$ since the user is allocated a sub-band shared by all base stations.

\begin{theorem}[Strict FFR, interior user]
\label{thm:FFRi}
The coverage probability of the interior user with Strict FFR is 
\[\bar{\mathrm{F}}_\mathrm{FFR,i}(T))=\frac{p_c(\max\{T,\TR\},\lambda,\alpha,1)}{p_c(\TR,\lambda,\alpha,1)},\]
and $p_c(T,\lambda,\alpha,\Delta)$ is given by (\ref{eq:pc}).
\end{theorem}

\begin{proof}
Starting from (\ref{eq:coverage}) and applying Bayes' rule, 
\begin{equation}
\bar{\mathrm{F}}_{\mathrm{FFR,i}}(T)=\mathbb{P}\left(\frac{P\hat{g_y} r^{-\alpha}}{\sigma^2+P{I_r}} > T ~\bigg|~ \frac{Pg_yr^{-\alpha}}{\sigma^2+P I_r} > \TR \right) \nonumber
\end{equation}
\begin{eqnarray}
&=&\frac{\mathbb{P}\left(\frac{P\hat{g_y} r^{-\alpha}}{\sigma^2+P I_r} > \max\{T,\TR\}\right)}{\mathbb{P}\left( \frac{Pg_yr^{-\alpha}}{\sigma^2+P I_r} > \TR \right)} \nonumber \\
&=& \frac{p_c(\max\{T,\TR\},\lambda,\alpha,1)}{p_c(\TR,\lambda,\alpha,1)}. \nonumber
\end{eqnarray}
\end{proof}

The $\max\{T,\TR\}$ term in the numerator is a result of interior users having $\sinr \geq \TR$ by definition. Also, because the interior users of all the cells share the same sub-band, the $\sinr$ CCDF is closely related to the results of \cite{Andrews2010} for users with no frequency reuse.

\subsection{SFR}
There are two major differences between SFR and Strict FFR. One is the use of power control, rather than frequency reuse for the edge users, controlled by the design parameter $\beta$. Additionally, the base stations can reuse all sub-bands, but apply $\beta$ to only one of the $\delta$ sub-bands. The interference power is given by (\ref{eq:SFRIR}), thus the interference term is denoted as $\eta P I_r$ where $\eta = (\Delta - 1  + \beta)/\Delta$ is the effective interference power factor, consolidating the impact of interference from the higher and lower power downlinks. We now consider the CCDF for SFR starting again with edge users, followed by the interior user case. 

\begin{theorem}[SFR, edge user]
\label{thm:SFRe}
The coverage probability of an SFR edge user whose initial $\sinr$ is less than $\TR$ is 
\[\bar{\mathrm{F}}_{\mathrm{SFR,e}}(T)=\frac{p_c(\frac{\eta T}{\beta},\lambda,\alpha,\Delta) - \pi \lambda \int_0^\infty e^{-\pi \lambda v (1 + 2\xi(T,\TR,\alpha,\Delta,\eta,\beta)) -\mu \eta(\frac{T}{\beta}+\TR) \frac{\sigma^2}{\beta P} v^{\alpha/2}} \dd v}{1-p_c(\TR,\lambda,\alpha,\Delta)},\]
where 
\[\xi(T,\TR,\alpha,\Delta,\eta,\beta)= \int_1^\infty\left[1- \frac{1}{1+\eta\TF x^{-\alpha}}\left(\frac{1}{1+\frac{\eta T}{\beta} x^{-\alpha}}\right) \right] x \dd x.\]
\end{theorem}

\begin{proof}
The proof is given in Appendix B.
\end{proof}
The expressions only differ slightly from Strict FFR. The inclusion of $\eta$ and $\beta$ can be viewed as creating effective $\sinr$ and FFR thresholds $\frac{\eta T}{\beta}$ and $\eta\TR$ respectively.

For SFR, the CCDF of the interior user, $\bar{\mathrm{F}}_{\mathrm{i}}(T)$ is found in the same manner as Strict FFR.
\begin{theorem}[SFR,interior user]
\label{thm:SFRi}
\[\bar{\mathrm{F}}_{\mathrm{SFR,i}}(T)=\frac{p_c(\eta\max\{T,\TR \},\alpha,1)}{p_c(\eta \TR,\alpha,1)}.\]
\end{theorem}
\begin{proof}
Follows from the definition of $\mathrm{F_{\mathrm{SFR},i}}(T)$, and Theorem \ref{thm:FFRi}.
\end{proof}
Again, the CCDF is similar in structure to the Strict FFR case. Also, since for interior users there is no extra $\beta$ power control in their transmit power, only the effective interference power factor $\eta$ remains in the expressions.

\section{\label{sec:discussion}Discussion of the Model}
In this section we present a special case where the coverage probability results of Section \ref{sec:coverage} can be significantly further simplified. As we will see, this allows much clearer insight into the performance of cell-edge users, something not previously possible with the grid model. We also provide a discussion of intuitive lower and upper bounds of the $\sinr$ distribution for edge users under SFR as well as a numerical comparison with the grid model and actual base station deployments for the different reuse strategies.

\subsection{\label{sec:FFRclosed}Strict FFR: No-noise and $\alpha=4$}
In the case of $\alpha=4$ and no noise, closed form expressions can be derived for the edge user coverage probability. In the case of Strict FFR, when $T\neq \TR$,
\begin{eqnarray}
\label{eq:FFRsimp}
\bar{\mathrm{F}}_\mathrm{FFR,e}(T) = \frac{1 + \rho(\TR,4)}{\rho(\TR,4)}\left(\frac{1}{1 + \frac{1}{\Delta}\rho(T,4)} - \frac{1}{1 + 2\xi(T,\TR,\lambda,4,\Delta)}\right)
\end{eqnarray}
where
\begin{eqnarray}
\xi(T,\TR,\lambda,4,\Delta) &=& \frac{2T^{3/2}\arctan(\frac{1}{\sqrt{T}})-\left(2\arctan(\frac{1}{\sqrt{\TR}})\right)\left({\TR}^{3/2}\Delta -\left(T\sqrt{\TR}\right)\left(1+\Delta\right)\right)}{4\Delta(\TR-T)}-\nonumber\\&~&\frac{\pi(T^{3/2}+{\TR}^{3/2})\Delta+(\pi T\sqrt{\TR})(1-\Delta)}{4\Delta(\TR-T)}.
\end{eqnarray}
We also note that for $\alpha=4$ and no noise, $\rho(T,4)$ has a closed form as well \cite{Andrews2010}. However, when $T =\TR$, the expression has an indeterminate form. By evaluating the limit $T \to \TR$, this simplifies to
\begin{align}
\xi(T,\TR,\lambda,4,\Delta) = \frac{\sqrt{\TR}\pi(2 \Delta +1)}{8\Delta}+\frac{2 \TR}{8(\TR+1)\Delta}-\frac{ \sqrt{\TR}(2\Delta+1)\cot^{-1}\left(\sqrt{\TR}\right)}{4\Delta}.
\end{align}
As a result of the assumptions, corresponding to interference-limited, urban cellular networks \cite{Ghosh2010}, we see that the $\sinr$ distribution of Strict FFR edge users are simply a function of the $\sinr$ threshold $T$ and the reuse threshold $\TR$. The reuse threshold determines whether a user is switched to a reuse-$\Delta$ sub-band. Although not given here, it is clear that the same applying this special case to the interior users would result in similarly closed form and simple expressions.

\subsection{\label{sec:SFRclosed}SFR: No-noise and $\alpha=4$}
Likewise for the SFR case,
\begin{eqnarray}
\label{eq:SFRsimp}
\bar{\mathrm{F}}_\mathrm{SFR,e}(T) = \frac{1 + \rho(\eta \TR,4)}{\rho(\eta \TR,4)}\left(\frac{1}{1 + \rho(\frac{\eta}{\beta} T,4)} - \frac{1}{1 + 2\xi(T,\TR,\lambda,4,\Delta)}\right)
\end{eqnarray}
where
\begin{eqnarray}
\xi(T,\TR,\lambda,4,\Delta) &=& \frac{{\eta}^{3/2}T\beta}{4\sqrt{\TR}(T-\TR\beta)}-\frac{\eta\beta{T}^3\left(2\arctan\left(\sqrt{\frac{\beta}{\eta T}}\right)+\pi\right)}{(T-\TR\beta)}+\nonumber\\
&~&\frac{\eta {T}^{3/2}{\TR}^{3/2}\beta^{5/2}\left(2\arctan\left(\frac{1}{\sqrt{\eta \TR}}\right)-\pi\right)}{(T-\TR\beta)}.
\end{eqnarray}

When $\TR = T$, the limit $\TR \to T$ simplifies as,
\begin{align}
\xi(T,\TR,\lambda,4,\Delta) &=& \frac{\sqrt{\eta\TR}\left(2\arctan\left(\frac{\beta}{\sqrt{(\eta(\TR)\beta)}}\right)-\pi\right)}{{4\sqrt{\beta}\left(\beta-1\right)}}-\frac{2\beta\sqrt{\eta\TR}\arctan\left(\frac{1}{\sqrt{\eta\TR}}\right)+\pi}
{4\left(\beta-1\right)}.
\end{align}
Once the again, the $\sinr$ distribution in \eqref{eq:SFRsimp} is only a function of the $\sinr$ threshold $T$ and in this case, the two SFR design parameters, the reuse threshold $\TR$ and the power control factor $\beta$, which influences the effective interference factor $\eta$. In the next several sections, we will exploit this simple structure to compare Strict FFR and SFR with other reuse strategies and evaluate their relative performance as a function of the design parameters.

\subsection{Comparison with No reuse / Standard Frequency Reuse}
Comparing the $\sinr$ distributions derived in Section \ref{sec:FFRclosed} and \ref{sec:SFRclosed} for edge users with those of a no-frequency reuse and standard reuse-$\Delta$ system as derived in \cite{Andrews2010} provides insight into the relative merits and tradeoffs associated with FFR. In Fig. \ref{fig:Pcomp} we plot the four systems with $\sigma^2=0$, $\alpha=4$, and $\TR = 1$ dB and compare the analytical expressions with Monte-Carlo simulations. Intuition for these results can be seen from the proofs of Theorems \ref{thm:FFRe} and \ref{thm:SFRe}. Since the probability is degraded multiplicatively by the interfering downlinks, the number of sources of interference drives the outage. When $\Delta=3$, we see that only 33\% of the base stations causing downlink interference to edge users under universal frequency reuse are active on the same resources under Strict FFR. However, SFR allows adjacent base stations to serve interior users on the same sub-bands used by adjacent edge users, increasing interference and lowering coverage. The reason for the sharp cutoffs in the coverage curves for no reuse and reuse-$\Delta$ users in Fig. \ref{fig:Pcomp} is because unlike FFR, edge users under those strategies do not get allocated a new sub-band and by definition their $\sinr$ must be below the reuse threshold $\TR$.

\subsection{Lower and Upper Bounds for SFR}
One question arising from Theorem \ref{thm:SFRe} is how the power control factor $\beta$ influences the distribution of the $\sinr$. Fig. \ref{fig:sfr} compares the coverage probability of SFR for increasing $\beta$ factors with Strict FFR. As $\beta$ increases, SFR performance approaches and then surpasses Strict FFR when $\beta \geq 15$, which is equivalent to a 12 dB power increase for cell-edge downlinks over interior user downlinks.    

Next, we show that the performance of an SFR system is bounded by two other reuse systems, namely reuse-1 (a.k.a no frequency reuse) when $\beta \to 1$ and reuse-$\Delta$ when $\beta \to \infty$.

\subsubsection{$\beta \rightarrow 1$}
As $\beta \rightarrow 1$ then $\eta \rightarrow 1$ as well and the CCDF of the edge user $\sinr$ is given as
\begin{eqnarray}
\bar{\mathrm{F}}_{\textrm{SFR},e}(T) &=& \mathbb{P}\left(\frac{P\hat{g_y} r^{-\alpha}}{\sigma^2+ P \hat{I_r}} > T ~\bigg|~ \frac{Pg_yr^{-\alpha}}{\sigma^2+P I_r} < \TR \right).
\end{eqnarray}
This is the same as a no frequency reuse strategy where a user with $\sinr \leq \TR$ is given a new frequency sub-band. The only benefit would be from the fading and random placement of the base stations, but the number of interfering base stations would be the same. 

\subsubsection{$\beta \rightarrow \infty$}
As $\beta \rightarrow \infty$ then $\eta \rightarrow \frac{1}{3}$. From \eqref{eq:SINR} this means the $\sinr$ of the inner user is 0 (in linear units), while the $\sinr$ for the edge user must be evaluated using L'Hopital's rule and is
\begin{equation}
\sinr = \frac{g_yr^{-\alpha}}{\displaystyle\sum_{e\in \mathcal{E}}{g_e}{R_{e}}^{-\alpha}}.
\end{equation}
However, because only $\frac{1}{\Delta}$ of the base stations are using $\beta$ for their transmit power, and are randomly chosen from the realization of the PPP with total density $\lambda$, this can be equivalently thought of as having interference from a thinned distribution with $\hat{\lambda}=\frac{\lambda}{\Delta}$. Thus we can utilize the result from \cite{Andrews2010} which showed that this is the same as the reuse-$\Delta$ case and \eqref{eq:coverage} applies.

Fig. \ref{fig:sfr} compares the computed lower and upper bounds with simulated SFR systems utilizing $\beta=1$ and $\beta$ approaching $\infty$ respectively and shows that the bounds are quite tight in both cases.

\subsection{Comparison with Grid Model}
As noted previously, the majority of work on the design of systems using fractional frequency reuse has focused on utilizing a grid model. The distance to the BS was used to classify the edge and interior users determine resource allocation strategies for the FFR sub-bands. In this section we compare the coverage results obtained using the spatial Poisson model with a uniformly spaced rectangular grid of base stations as well as with simulations utilizing the base station locations of an urban deployment by a major service provider.

In Figs. \ref{fig:expFFR} and \ref{fig:expSFR} we compare the CCDFs of the $\sinr$ obtained using the PPP model with distributions obtained using a grid model as well as an actual base station deployment for Strict FFR and SFR respectively. The grid model, as expected, is more optimistic in terms of coverage probability than the results based on the actual deployment \cite{Brown2000}, \cite{Andrews2010}. This is primarily due to the minimum distance between the interfering base stations and the typical edge user, resulting in well-defined fixed-sized tiers of interference, with the outlying tiers much less important to calculating the overall performance due to the exponentially decaying nature of the pathloss. Thus from the geometry of the network, we see that Strict FFR provides better coverage than SFR to edge users since the dominant interfering downlinks originate from the first tier of base stations surrounding a cell, and under Strict FFR, none of these base stations are contributing interference. 

With the PPP model, the number of interfering base stations within a region, \ie the size of a cellular tier, is random and the distances are not lower bounded, except for the fact that the edge user is assumed to be closer to the serving base station than any of the interfering base stations. However, despite this difference, the distributions for SFR follow a similar sloping shape, but those for Strict FFR do not. In this case the gap may be attributed to the lack of a fixed reuse plan in the Poisson model. The comparisons of Figs. \ref{fig:expFFR} and \ref{fig:expSFR} also verify the claim that the PPP model serves as a lower bound on performance in a real deployment.

\section{\label{sec:rate} Rate and Resource Allocation for FFR}
Systems under various traffic loads and channel conditions may have different priorities in regards to which metrics are most important. For example, networks experiencing high traffic loads may wish to optimize spectral efficiency, however in another circumstance, providing peak data rates for interference-limited edge users may be the desired goal. Since optimizing one metric for FFR systems usually leads to sub-optimal performance in regards to the other metrics, designers may additionally consider a hierarchy for the tradeoffs, by fixing thresholds for multiple metrics and optimizing the remaining ones in order to compromise between improving throughput and maintaining resource efficiency. This section explores these tradeoffs and compares the performance of SFR and FFR using the Poisson model and proposes a resource allocation strategy utilizing the analytical $\sinr$ distributions for maximizing sum-rate and balancing resource efficiency based on traffic load.

\subsection{Average Edge User Rate}
In modern cellular networks utilizing OFDMA, user rate is directly related to average $\sinr$ and the system's resource allocation algorithm. Again as with the $\sinr$ distributions, most prior work utilizing the grid model relied on simulations to analyze the performance. The coverage results derived in Section \ref{sec:coverage} can be extended to develop average user rate expressions under Strict FFR or SFR, creating a new set of system design tools for general and hybrid FFR strategies. Additionally, this would allow for greater insight into the joint optimization of coverage and rate. 

Adaptive modulation and coding is assumed such that users are able to achieve the average data rate and the expressions are given in terms of nats/Hz, where 1 bit $= \log_e(2)$ nats. We define the average rate of a edge user to be 
\begin{equation}
\label{eq:rate} 
\bar{\tau} = \mathbb{E}[\ln(1+\sinr)],
\end{equation}
averaging over the base station locations and the fading distributions \cite{Andrews2010}.

\subsubsection{Strict FFR}
First we consider the typical edge user given a Strict FFR sub-band. 
\begin{theorem}[Strict FFR, edge user]
\label{thm:RateFFR}
The average rate of an edge user under Strict FFR is 
\begin{equation}
\tau(\TR,\lambda,\alpha) = \int_{t > 0}\frac{p_c\left((e^{t}-1),\alpha,\Delta\right) - \xi\left((e^{t}-1),\TR,\alpha,\Delta\right)}{1-p_c(\TR,\alpha,1)}\dd t
\end{equation}
where 
\[\xi\left((e^{t}-1),\TR,\alpha,\Delta\right)= \int_1^\infty\left[1- \frac{1}{1+\TR x^{-\alpha}}\left(1-\frac{1}{\Delta}\left(1-\frac{1}{1+(e^{t}-1) x^{-\alpha}}\right)\right) \right] x \dd x\]
and $p_c\left((e^{t}-1),\alpha,\Delta\right)$ is given by \eqref{eq:pc}.
\end{theorem}

\begin{proof}
The proof is given in Appendix C.
\end{proof}

These results are clearly related to the coverage probability for Strict FFR given in Theorem \ref{thm:FFRe}. As a result, these results can be evaluated using numerical integration and reduce to simple expressions for the same special cases as presented in Section \ref{sec:coverage}. 

\subsubsection{SFR}
The case for SFR edge users follows similarly to Theorem \ref{thm:RateFFR}.
\begin{theorem}[SFR, edge user]
\label{thm:RateSFR}
The average rate of an edge user under SFR is 
\begin{equation}
\tau(\TR,\lambda,\alpha,\beta,\eta) = \int_{t > 0}\frac{p_c\left(\frac{\eta}{\beta}(e^{t}-1),\alpha,\Delta\right) - \xi\left(\frac{\eta}{\beta}(e^{t}-1),\eta\TR,\alpha,\Delta\right)}{1-p_c(\eta\TR,\alpha,1)}\dd t
\end{equation}
where 
\[\xi\left(\frac{\eta}{\beta}(e^{t}-1),\eta\TR,\alpha,\Delta\right)= \int_1^\infty\left[1- \frac{1}{1+\eta\TR x^{-\alpha}}\left(\frac{1}{1+\frac{\eta}{\beta}(e^{t}-1) x^{-\alpha}}\right) \right] x \dd x\]
and $p_c\left(\frac{\eta}{\beta}(e^{t}-1),\alpha,\Delta\right)$ is given by \eqref{eq:pc}. 
\end{theorem}
 
\begin{proof}
The proof is given in Appendix D.
\end{proof}

Fig. \ref{fig:Ecap} compares the average rates for edge users under Strict FFR and SFR with $\beta = 4$ with no reuse and reuse-3 as a function of the threshold $\TR$. We note that Strict FFR provides the highest average rates since it is also the strategy that provides the highest coverage for edge users. Also, the average rate increases linearly as $\TR$ is increased because users with increasingly higher initial $\sinr$ are provided a FFR sub-band. As with coverage probability, as $\beta$ increases, edge users under SFR can have a higher rate than Strict FFR. However, since $\eta$ also increases with $\beta$, this gain in rate for edge users is a tradeoff with decreasing average rates for interior users.

\subsection{Resource Allocation for FFR}
Much of the research on FFR system design has focused on how to determine the size of the frequency partitions \cite{Assad2008}, \cite{Novlan2010}, \cite{Alsawah2008}, \cite{Chen2010}. For example, in a typical LTE system with a bandwidth of 10 MHz, 50 sub-bands may be available to serve users per cell, each one with a bandwidth of 200 kHz \cite{Ghosh2010}. Given a reuse factor of $\Delta$ and $N_{\rm band}$ total sub-bands available to the cell, the allocation of sub-bands for interior users $N_{\rm int}$ and edge users $N_{\rm edge}$ is given as
\begin{equation}
N_{\rm edge} = \left\lfloor (N_{\rm band}-N_{\rm int})/\Delta \right\rfloor
\label{eq:StrictAllocation}
\end{equation}
For SFR, all the sub-bands are reallocated in the cell, although the partitioning of sub-bands between edge and interior users is given as
\begin{eqnarray}
\label{eq:SFRAllocation}
N_{\rm int} &=& N_{\rm band}-N_{\rm edge} \\
\textrm{where} ~ N_{\rm edge} &\leq& N_{\rm band}/\Delta 
\end{eqnarray}
From these equations we note that one of the advantages of SFR over Strict FFR is the ability to achieve 100\% allocation unlike Strict FFR, due to the sharing of resources between interior and edge users. However, as seen in Section \ref{sec:discussion}, this results in a tradeoff between $\sinr$ improvement for edge users and network spectral efficiency.
 
\subsection{System Design Guidelines}
This section gives system design guidelines for SFR and Strict FFR based on the analytical $\sinr$ distributions and are verified by Monte-Carlo simulations. The total number of sub-bands in the system under consideration is 48, comparable to a 10MHz LTE deployment. The user snapshots are taken over a 10 $\textrm{km}^2$ area with 25 uniformly spaced grid base stations, while the PPP base stations are modeled with a corresponding density of $\lambda = 1/(4000\pi^2)$ base stations per $\textrm{m}^2$.

\textbf{Coverage}. From the shape of the curves in Fig. \ref{fig:Pcomp} it is noted that at low values of $\beta$, SFR provides lower coverage probability compared to Strict FFR. However as seen in Section \ref{sec:discussion}, if $\beta$ is sufficiently large, SFR will surpass Strict FFR in terms of coverage as it approaches the performance of reuse-$\Delta$. This tradeoff is achieved when there is approximately a 12 dB difference between downlink transmit power to edge and interior users. Increasing $\beta$ beyond this results in diminishing performance gain compared to the substantially increased required transmit powers.  

\textbf{Spectral Efficiency}. However, under high traffic loads, interference avoidance may not outweigh the cost of reserving bandwidth for the partitioning structures of the reuse systems, especially reuse-$\Delta$, or Strict FFR with a high number of sub-bands allocated to edge users. Reduction in resource efficiency additionally hurts the peak throughput of the cell, since users with high rate requirements may not be able to be allocated sufficient number of sub-bands. One benefit of SFR is the ability to balance the $\sinr$ gains experienced under Static FFR while utilizing more of the available sub-bands in every cell.

\textbf{Sum Rate}. Network sum rate performance was evaluated by running simulations of the various systems using $\TR = 3$ dB. The number of sub-bands available to edge users was varied from 2 - 16, representing the maximum number of edge user sub-bands since $\Delta = 3$ and $N_{band}/3 = 16$. This is analogous to varying the interior radius of FFR systems under the grid model \cite{Novlan2010}. Fig. \ref{fig:rSumcap} indicates that SFR with $\beta = 2$ is able to provide higher sum-rate than standard systems without sensitivity to the number of sub-bands allocated because when $N_{ext} \leq N_{band}/\Delta$, all the available sub-bands are allocated. However, in the case of Strict FFR, as the number of sub-bands is increased, the total sum rate of Strict FFR decreases. This is because of Strict FFR's fundamental tradeoffs between spectral efficiency and the improved performance provided to edge users. Also note that when $N_{ext} = N_{band}/\Delta$, Strict FFR does not converge to reuse-$\Delta$. Instead it has lower sum-rate. This is because although both systems allocate the same number of sub-bands in this case, reuse-$\Delta$ gives resources to interior and edge users, while Strict FFR only allocates resources to edge users, who by definition have smaller received power due to path loss and interference, reducing the achievable rate.

\subsection{\label{sec:algorithm} $\sinr$-Proportional Resource Allocation}
Under the grid model, the frequency reuse partitions are based on the geometry of the network, and the resource allocation between interior and cell-edge users is proportional to the square of the ratio of the interior radius and the cell radius. In the case of the Poisson based model, geometric intuition for sub-band allocation does not apply, and instead allocation should be made based on $\sinr$ distributions from Section \ref{sec:coverage}, improving the sum-rate for Strict FFR and SFR over that shown in Fig. \ref{fig:rSumcap}.

Based on a chosen FFR threshold $\TR$, the number of sub-bands can then be chosen by evaluating the CCDF at $\TR$ and choosing $N_{edge}$ to be proportional to that value. In other words,
\begin{equation}
N_{edge} = \lfloor \left(1-\bar{\mathrm{F}}_{FFR,e}(\TR)\right)N_{band} \rfloor.
\end{equation}

The threshold $\TR$ may be set as a design parameter, or may be alternatively chosen based on traffic load by inverting the CCDF (\ie low $\TR$ represents low edge user traffic and high $\TR$ when there are a large number of edge users). Fig. \ref{fig:tSumcap} presents the results of simulations of this $\sinr$-proportional algorithm as function of $\TR$. Both SFR and Strict FFR outperform the standard reuse strategies. SFR outperforms Strict FFR for smaller values of $\TR$, due to the the loss in spectral efficiency of Strict FFR. As $\TR$ increases, Strict FFR provides greater sum-rate, due to larger gain in coverage for edge users when the number of allocated sub-bands for edge users is large.

\section{\label{sec:conclusion} Conclusion}
This work has presented a new analytical framework to evaluate coverage probability and average rate in Strict FFR and SFR systems leading to tractable expressions. The resulting system design guidelines highlighted the merits of those strategies as well as the tradeoffs between the superior interference reduction of Strict FFR and the greater resource efficiency of SFR. A natural extension of this work is to address the cellular uplink. While many of the same takeaways would be expected, the inclusion of fine-granularity power control and the metric of total power consumption, which is especially important for battery powered user devices \cite{Wamser2010}, would be expected to have a significant factor on the results.

Additionally, this work motivates future research using the Poisson model to evaluate ICIC strategies using base station cooperation to allow FFR to be implemented dynamically alongside resource allocation strategies to adapt to different channel conditions and user traffic loads in each cell \cite{Hamouda2009,Chen2008}. A cohesive framework would allow for research into the dynamics and implications of FFR along with other important cellular network research including handoffs, base station cooperation, and FFR in conjunction with relays and/or femtocells. 

\appendices
\section{\label{sec:AppA} Proof of Theorem 1}
A user $y$ with $\sinr < \TR$ is given a FFR sub-band $\delta_y$, where $\delta \in \{1,...,\Delta\}$ with uniform probability $\frac{1}{\Delta}$, and experiences new fading power $\hat{g_y}$ and out-of-cell interference $\hat{I_r}$, instead of $g_y$ and $I_r$. The CCDF of the edge user $\bar{\mathrm{F}}_{\mathrm{FFR,e}}(T)$ is now conditioned on its previous $\sinr$,
\begin{equation}
\bar{\mathrm{F}}_{\mathrm{FFR,e}}(T)=\mathbb{P}\left(\frac{P\hat{g_y} r^{-\alpha}}{\sigma^2+P \hat{I_r}} > T ~\bigg|~ \frac{Pg_yr^{-\alpha}}{\sigma^2+P I_r} < \TR \right).
\end{equation}
Using Bayes' rule we have,
\begin{equation}
\bar{\mathrm{F}}_{\mathrm{FFR,e}}(T)=\frac{\mathbb{P}\left(\frac{P\hat{g_y} r^{-\alpha}}{\sigma^2+P \hat{I_r}} > T \ , \frac{Pg_yr^{-\alpha}}{\sigma^2+P I_r} < \TR   \right)}{\mathbb{P}\left( \frac{Pg_yr^{-\alpha}}{\sigma^2+P I_r} < \TR \right)}.
\end{equation}
Since $\hat{g_y}$ and $g_y$ are i.i.d. exponentially distributed with mean $\mu$, this gives
\begin{equation}
\mathbb{P}\left(\frac{P\hat{g_y} r^{-\alpha}}{\sigma^2+P \hat{I_r}} > T \ , \frac{Pg_yr^{-\alpha}}{\sigma^2+P I_r} < \TR \right) = \E \left[e^{\left(-\mu \frac{T}{P} r^{\alpha}(\sigma^2 +P\hat{I_r}) \right)}\left(1-e^{\left(-\mu \frac{\TR}{P} r^{\alpha}(\sigma^2 +P{I_r}) \right)}\right)\right],
\end{equation}
which equals
\[p_c(T,\lambda,\alpha,\Delta) - \E \left[e^{\left(-\mu \frac{T}{P} r^{\alpha}(\sigma^2 +P \hat{I_r}) \right)}e^{\left(-\mu \frac{\TR}{P} r^{\alpha}(\sigma^2 +P {I_r}) \right)}\right],\]
where $p_c(T,\lambda,\alpha,\Delta)$ is obtained in (\ref{eq:pc}). Hence
\[\bar{\mathrm{F}}_{\mathrm{FFR,e}}(T)=\frac{p_c(T,\alpha,\Delta) -\E \left[e^{ \left(-\mu r^\alpha \frac{\sigma^2}{P}(T +\TR)\right)} e^{\left(-\mu  r^{\alpha}(T\hat{I_r}+\TR I_r) \right)}\right] }{1-p_c(\TR,\alpha,1)}.\]
Now concentrating on the second term and {\em conditioning on $r$}, \ie the distance to the nearest BS, we observe that the expectation of $ \exp\left(-\mu  r^{\alpha}(T\hat{I_r}+\TR I_r) \right)$ is the joint Laplace transform of $\hat{I_r}$ and $I_r$ evaluated at $(\mu  r^{\alpha}T, \mu  r^{\alpha}\TR)$. The joint Laplace transform is
\begin{align*}
\mathscr{L}(s_1,s_2)&= \E \exp \left(-s_1\hat{I_r}-s_2 I_r \right)\\
&= \E \exp \left(-s_1\sum_{z\in Z} G_z R_z^{-\alpha}\mathbf{1}(\delta_z = \delta_y) -s_2\sum_{z\in Z} g_z R_z^{-\alpha}\right)\\
&= \E \exp \left({-\sum_{z\in Z} (s_1G_z R_z^{-\alpha}\mathbf{1}(\delta_z = \delta_y) +s_2g_z R_z^{-\alpha})}\right)\\
&= \E \prod_{z\in Z}\exp\left(-(s_1G_z R_z^{-\alpha}\mathbf{1}(\delta_z = \delta_y) +s_2g_z R_z^{-\alpha})\right)\\
&= \E \prod_{z\in Z} \exp(-s_2g_z R_z^{-\alpha})\left(1-\E\left(\mathbf{1}(\delta_z = \delta_y)\right)(1-\exp(-s_1G_z R_z^{-\alpha}))\right),
\end{align*}
where $\mathbf{1}(\delta_y = \delta_z)$ is an indicator function that takes the value 1, if base station $z$ is transmitting to an edge user on the same sub-band $\delta$ as user $y$.\\
Since $G_z$ and $g_z$ are also exponential with mean $\mu$, we can evaluate the above expression as
\begin{equation}
\E \prod_{z\in Z}\frac{\mu}{\mu+s_2R_z^{-\alpha}}\left(1-\frac{1}{\Delta}\left(1-\frac{\mu}{\mu+s_1R_z^{-\alpha}}\right) \right).
\label{eq:LFFR}
\end{equation}
By using the probability generating functional (PGFL) of the PPP \cite{Stoyan1996}, we obtain the Laplace transform as
\[\mathscr{L}(s_1,s_2)= \exp\left(-2\pi \lambda \int_r^\infty\left[1- \frac{\mu}{\mu+s_2 x^{-\alpha}}\left(1-\frac{1}{\Delta}\left(1-\frac{\mu}{\mu+s_1 x^{-\alpha}}\right) \right) \right] x \dd x \right).\]
Hence
\[\mathscr{L}(\mu r^\alpha T, \mu r^\alpha \TR )= \exp\left(-2\pi \lambda r^2 \int_1^\infty\left[1- \frac{1}{1+\TR    x^{-\alpha}}\left(1-\frac{1}{\Delta}\left(1-\frac{1}{1+T x^{-\alpha}}\right) \right) \right] x \dd x \right).\]
De-conditioning on $r$, we have 
\[\E e^{\left(-\mu \frac{T}{P} r^{\alpha}(\sigma^2 +P \hat{I_r}) \right)}e^{\left(-\mu \frac{\TR}{P} r^{\alpha}(\sigma^2 +P {I_r}) \right)} = \pi \lambda \int_0^\infty e^{-\pi \lambda v (1 +2\xi(T,\TR,\Delta,\alpha)) -\mu (T+\TR) \frac{\sigma^2}{P} v^{\alpha/2}} \dd v, \]
where 
\[\xi(T,\TR,\alpha,\Delta)= \int_1^\infty\left[1- \frac{1}{1+\TR x^{-\alpha}}\left(1-\frac{1}{\Delta}\left(1-\frac{1}{1+T x^{-\alpha}}\right) \right) \right] x \dd x.\]

\section{\label{sec:AppB} Proof of Theorem 3}
A user $y$ with $\sinr < \TR$ is given a SFR sub-band $\delta$, where $\delta \in \{1,...,\Delta\}$, transmit power $\beta P$ and experiences new fading power $\hat{g_y}$ and out-of-cell interference $\hat{I_r}$, instead of $g_y$ and $I_r$. The CCDF of the edge user $\bar{\mathrm{F}}_{\mathrm{SFR,e}}(T)$ is now conditioned on its previous $\sinr$,
\begin{equation}
\bar{\mathrm{F}}_{\mathrm{SFR,e}}(T)=\mathbb{P}\left(\frac{\beta P\hat{g_y} r^{-\alpha}}{\sigma^2+\eta P \hat{I_r}} > T ~\bigg|~ \frac{Pg_yr^{-\alpha}}{\sigma^2+\eta P I_r} < \TR \right).
\end{equation}
Using Bayes' rule we have,
\begin{equation}
\bar{\mathrm{F}}_{\mathrm{SFR,e}}(T)=\frac{\mathbb{P}\left(\frac{\beta P\hat{g_y} r^{-\alpha}}{\sigma^2+\eta P \hat{I_r}} > T \ , \frac{Pg_yr^{-\alpha}}{\sigma^2+\eta P I_r} < \TR \right)}{\mathbb{P}\left( \frac{Pg_yr^{-\alpha}}{\sigma^2+\eta P I_r} < \TR   \right)}.
\end{equation}
Since $\hat{g_y}$ and $g_y$ are exponentially distributed with mean $\mu$, we have
\begin{equation}
\mathbb{P}\left(\frac{\beta P\hat{g_y} r^{-\alpha}}{\sigma^2+\eta P \hat{I_r}} > T \ , \frac{Pg_yr^{-\alpha}}{\sigma^2+\eta P I_r} < \TR \right) = \E \left[ e^{\left(-\mu \frac{T}{\beta P} r^{\alpha}(\sigma^2 +\eta P\hat{I_r}) \right)}\left(1-e^{\left(-\mu \frac{\TR}{P} r^{\alpha}(\sigma^2 +\eta P{I_r}) \right)}\right)\right].
\end{equation}
and as before in Theorem 1,
\[\bar{\mathrm{F}}_{\mathrm{SFR,e}}(T)=\frac{p_c(\eta \frac{T}{\beta},\alpha,\Delta) -\E \exp \left(-\mu r^\alpha \frac{\sigma^2}{P}(\frac{T}{\beta} +\TR)\right) \exp\left(-\mu  r^{\alpha}\eta(\frac{T}{\beta}\hat{I_r}+\TR I_r) \right) }{1-p_c(\eta\TR,\alpha,1)}.\] 
Following the method of Theorem 1, concentrating on the second term and {\em conditioning on $r$}, \ie the distance to the nearest BS, we observe that the expectation of $ \exp\left(-\mu  r^{\alpha}\eta(\frac{T}{\beta}\hat{I_r}+\TR I_r) \right)$ is the joint Laplace transform of $\hat{I_r}$ and $I_r$ evaluated at $(\mu r^{\alpha}\eta\frac{T}{\beta}, \mu r^{\alpha}\eta\TR)$. The steps to evaluate the joint Laplace transform are the same as Theorem 1 with the exception that from \eqref{eq:SFRIR} we know that the structure of $\hat{I_r}$ and $I_r$ includes all base stations not just those associated with the user's sub-band $\delta$. This can equivalently thought of as setting $\Delta = 1$ in \eqref{eq:LFFR}. Thus,
\begin{equation}
\mathscr{L}(s_1,s_2) = \E \prod_{z\in Z}\frac{\mu}{\mu+s_2R_z^{-\alpha}}\left(\frac{\mu}{\mu+s_1R_z^{-\alpha}}\right).
\end{equation}
Using the PGFL, we obtain the Laplace transform as
\[\mathscr{L}(s_1,s_2)= \exp\left(-2\pi \lambda \int_r^\infty\left[1- \frac{\mu}{\mu+s_2 x^{-\alpha}}\left(\frac{\mu}{\mu+s_1 x^{-\alpha}}\right) \right] x \dd x \right).\]
Hence
\[\mathscr{L}\left(\mu r^\alpha \eta\frac{T}{\beta}, \mu r^\alpha \eta\TR \right)= \exp\left(-2\pi \lambda r^2 \int_1^\infty\left[1- \frac{1}{1+\eta\TR    x^{-\alpha}}\left(\frac{1}{1+\eta\frac{T}{\beta} x^{-\alpha}}\right) \right] x \dd x \right).\]
Finally, de-conditioning on $r$, we have 
\[\E \exp\left(-\mu  r^{\alpha}(\eta\frac{T}{\beta}\hat{I_r}+\eta\TR I_r) \right) = \pi \lambda \int_0^\infty e^{-\pi \lambda v (1 +2\xi(\eta\frac{T}{\beta},\eta\TR,\alpha,1)) -\mu (\eta\frac{T}{\beta}+\eta\TR) \frac{\sigma^2}{P} v^{\alpha/2}} \dd v, \]
where 
\[\xi(T,\TR,\alpha,1)= \int_1^\infty\left[1- \frac{1}{1+\eta\TR x^{-\alpha}}\left(\frac{1}{1+\eta\frac{T}{\beta} x^{-\alpha}}\right) \right] x \dd x.\]

\section{\label{sec:AppC} Proof of Theorem 5}
The average rate of an edge user, $\bar{\tau}(\TR,\lambda,\alpha)$, is determined by integrating over the $\sinr$ distribution derived in Theorem \ref{thm:FFRe}. Starting from \eqref{eq:rate} we have
\begin{eqnarray}
\bar{\tau}(\TR,\lambda,\alpha) &=& \mathbb{E}\left[\ln\left(1 + \sinr\right)\right] \nonumber\\
&=& \int_{r > 0}e^{-\pi\lambda r^2}\mathbb{E}\left[\ln\left(1 + \frac{P \hat{g_y} r^{-\alpha}}{\sigma^2+P \hat{I_r}}\right)\right]2\pi\lambda r\dd r, 
\end{eqnarray}
where we use the fact that since the rate $\tau = \ln(1 + \sinr)$ is a positive random variable, $E[\tau] = \int_{t>0} P(\tau > t)\dd t$. Following the approach of Theorem 1, we condition the edge user's new $\sinr$ based on the previous value, guaranteed to be below $\TR$ and have
\begin{equation}
\bar{\tau}(\TR,\lambda,\alpha) = \int_{r > 0}e^{-\pi\lambda r^2}\int_{t > 0} \mathbb{P}\left[\ln\left(1 + \frac{P\hat{g_y} r^{-\alpha}}{\sigma^2+P \hat{I_r}}\right) > t ~\bigg|~ \frac{Pg_yr^{-\alpha}}{\sigma^2+P I_r} < \TR\right]2\pi\lambda \dd t ~ r\dd r. \nonumber
\end{equation}
Applying Bayes' rule gives,
\begin{eqnarray}
&\mathbb{P}&\left[\ln\left(1 + \frac{P\hat{g_y} r^{-\alpha}}{\sigma^2+P \hat{I_r}}\right) > t ~\bigg|~ \frac{Pg_yr^{-\alpha}}{\sigma^2+P I_r} < \TR\right]\nonumber\\
&=&\frac{\mathbb{P}\left(\ln\left(1 + \frac{P\hat{g_y} r^{-\alpha}}{\sigma^2+P \hat{I_r}}\right) > t \ , \frac{Pg_yr^{-\alpha}}{\sigma^2+P I_r} < \TR\right)}{\mathbb{P}\left(\frac{Pg_yr^{-\alpha}}{\sigma^2+P I_r} < \TR \right)}
\end{eqnarray}
Following the method of Theorem \ref{thm:FFRe} gives
\begin{eqnarray}
\tau(\TR,\lambda,\alpha) &=& \int_{t > 0}2\pi\lambda\int_{r > 0}e^{-\pi\lambda r^2}\frac{\E \left[e^{-\mu r^\alpha \frac{e^{t}-1}{P}\left(\sigma^2 + P\hat{I_r}\right)}\right]}{1-\E \left[e^{ \left(-\mu r^\alpha \frac{\TR}{P}(\sigma^2 + PI_r)\right)}\right] }\dd t ~ r\dd r \nonumber\\ 
&~&-\int_{t > 0}2\pi\lambda\int_{r > 0}e^{-\pi\lambda r^2}\frac{\E \left[e^{ \left(-\mu r^\alpha \frac{\sigma^2}{P}(e^{t}-1 +\TR)\right)} e^{\left(-\mu  r^{\alpha}((e^{t}-1)\hat{I_r}+\TR I_r) \right)}\right]}{1-\E \left[e^{ \left(-\mu r^\alpha \frac{\TR}{P}(\sigma^2 + PI_r)\right)}\right] }\dd t~ r\dd r \nonumber\\
&=&\int_{t > 0}\frac{p_c(e^{t}-1,\alpha,\Delta) - \xi(e^{t}-1,\TR,\alpha,\Delta)}{1-p_c(\TR,\alpha,1)}\dd t,
\end{eqnarray}
where 
\[\xi(e^{t}-1,\TR,\alpha,\Delta)= \int_1^\infty\left[1- \frac{1}{1+\TR x^{-\alpha}}\left(1-\frac{1}{\Delta}\left(1-\frac{1}{1+(e^{t}-1) x^{-\alpha}}\right)\right) \right] x \dd x,\]
and $p_c\left(e^{t}-1,\alpha,\Delta\right)$ is given by \eqref{eq:pc}.
 
\section{\label{sec:AppD} Proof of Theorem 6}
Starting from \eqref{eq:rate} and integrating over the edge user $\sinr$ distribution for SFR derived in Theorem \ref{thm:SFRe} we have
\begin{eqnarray}
\tau(\TR,\lambda,\alpha) &=& \mathbb{E}\left[\ln\left(1 + \sinr\right)\right] \nonumber\\
&=& \int_{r > 0}e^{-\pi\lambda r^2}\mathbb{E}\left[\ln\left(1 + \frac{P\beta\hat{g_y} r^{-\alpha}}{\sigma^2+ P\eta\hat{I_r}}\right)\right]2\pi\lambda r\dd r \nonumber\\
&=& \int_{r > 0}e^{-\pi\lambda r^2}\int_{t > 0} \mathbb{P}\left[\ln\left(1 + \frac{P\beta\hat{g_y} r^{-\alpha}}{\sigma^2+P\eta \hat{I_r}}\right) > t ~\bigg|~ \frac{Pg_yr^{-\alpha}}{\sigma^2+P\eta I_r} < \TR\right]2\pi\lambda \dd t ~ r\dd r. \nonumber
\end{eqnarray}
Following the method of Theorem \ref{thm:SFRe} gives
\begin{eqnarray}
\tau(\TR,\lambda,\alpha) &=& 2\pi\lambda\int_{r > 0}e^{-\pi\lambda r^2}\int_{t > 0}\frac{\E \left[e^{ \left(-\mu r^\alpha \frac{e^{t}-1}{P\beta}(\sigma^2 + \eta\beta\hat{I_r})\right)}\right]}{1-\E \left[e^{\left(-\mu r^\alpha \frac{\TR}{P}\left(\sigma^2 + P\eta\beta I_r\right)\right)}\right] }\dd t ~ r\dd r \nonumber\\ 
&~&-2\pi\lambda\int_{r > 0}e^{-\pi\lambda r^2}\int_{t > 0}\frac{\E \left[e^{\left(-\mu r^\alpha \frac{\sigma^2}{P\beta}(e^{t}-1 + \beta\TR)\right)} e^{\left(-\mu r^{\alpha}\eta((e^{t}-1)\hat{I_r}+\TR I_r) \right)}\right]}{1-\E \left[e^{\left(-\mu r^\alpha \frac{\TR}{P}(\sigma^2 + P\eta\beta I_r\right)}\right] }\dd t ~ r\dd r \nonumber\\
&=&\int_{t > 0}\frac{p_c(\frac{\eta}{\beta}(e^{t}-1),\alpha,1)}{1-p_c(\eta\TR,\alpha,1)} - \frac{\xi(\frac{\eta}{\beta}(e^{t}-1),\eta \TR,\alpha,\Delta)}{1-p_c(\eta\TR,\alpha,1)}\dd t,
\end{eqnarray}
where 
\[\xi\left(\frac{\eta}{\beta}(e^{t}-1),\eta\TR,\alpha,\Delta \right)= \int_1^\infty\left[1- \frac{1}{1+\eta\TR x^{-\alpha}}\left(\frac{1}{1+\frac{\eta}{\beta}(e^{t}-1) x^{-\alpha}} \right) \right] x \dd x,\]
and $p_c\left(\frac{\eta}{\beta}(e^{t}-1),\alpha,\Delta\right)$ is given by \eqref{eq:pc}.

\linespread{1.5}
\bibliographystyle{IEEEtran}
\bibliography{ffr}

\newpage
\begin{figure}
\centering
\includegraphics[width=5in]{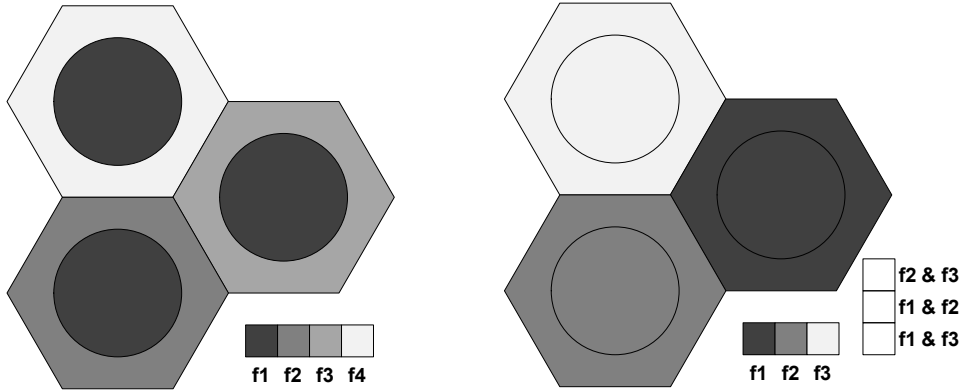}

\caption{Strict FFR (left) and SFR (right) deployments with $\Delta=3$ cell-edge reuse factor in a standard hexagonal grid model. The Poisson model maintains the resource partitions, but they are no longer of uniform geographical size or shape.}
\label{fig:FFRmodels}
\end{figure}

\begin{figure}
\centering
\includegraphics[width=5in]{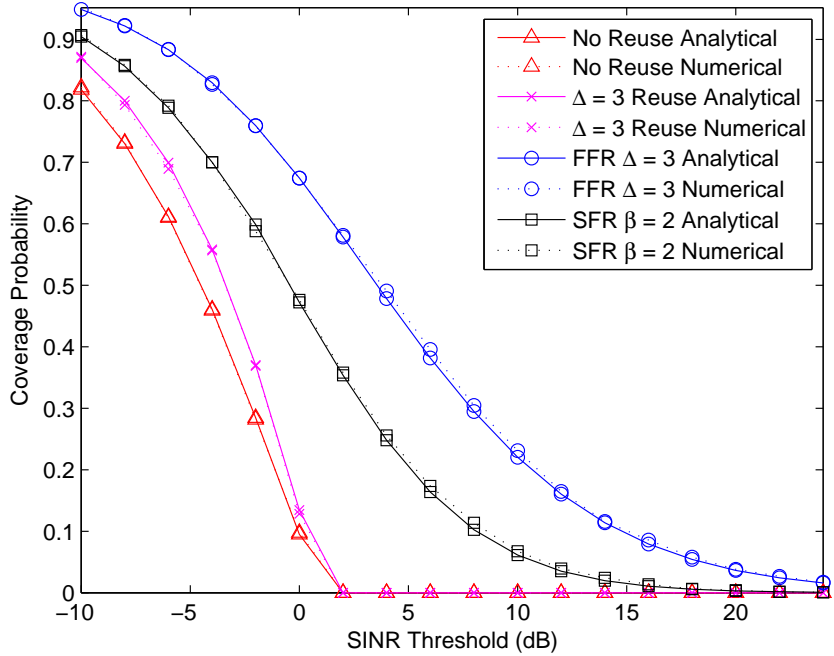}
\caption{Comparison of analytical coverage probability expressions for edge users using the Poisson model versus Monte-Carlo simulations with no noise, $\Delta=3$, $\TR=1$dB, and $\alpha=4$. The sharp cutoffs in the no reuse and reuse-3 curves are a result of those strategies not allocating a new sub-band to users with $\sinr$ below the coverage threshold $\TR$, unlike the FFR strategies.}
\label{fig:Pcomp}
\end{figure}


\begin{figure}
\centering
\includegraphics[width=5in]{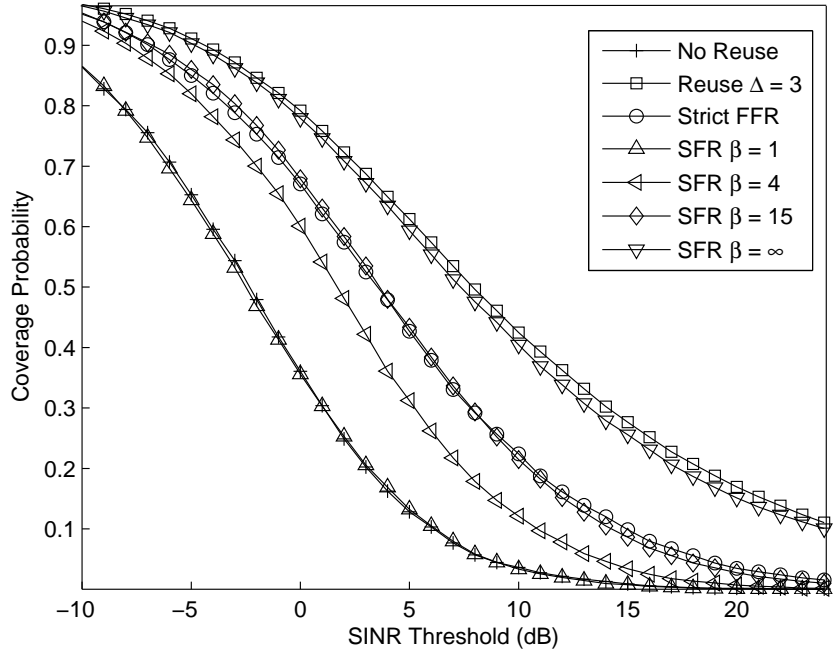}
\caption{Comparison of the $\sinr$ distribution of cell-edge users of SFR with different $\beta$ power-control factors, $\TR = 1$dB, no noise, and $\alpha=4$ with Strict FFR and the derived lower and upper bounds.}
\label{fig:sfr}
\end{figure}



\begin{figure}
\centering
\subfigure[Strict FFR]{
\includegraphics[width=3.25in]{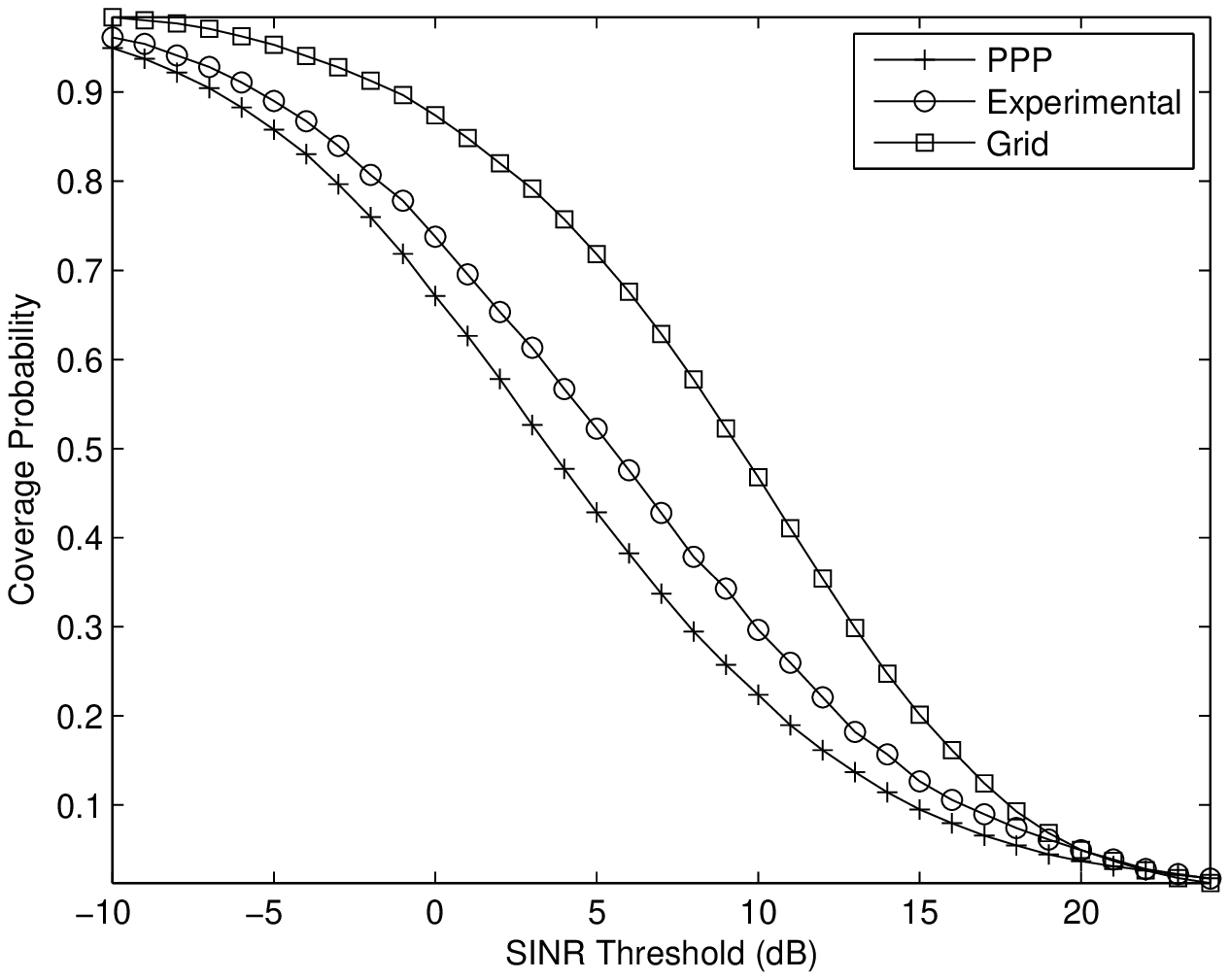}
\label{fig:expFFR}
}
\subfigure[SFR]{
\includegraphics[width=3.25in]{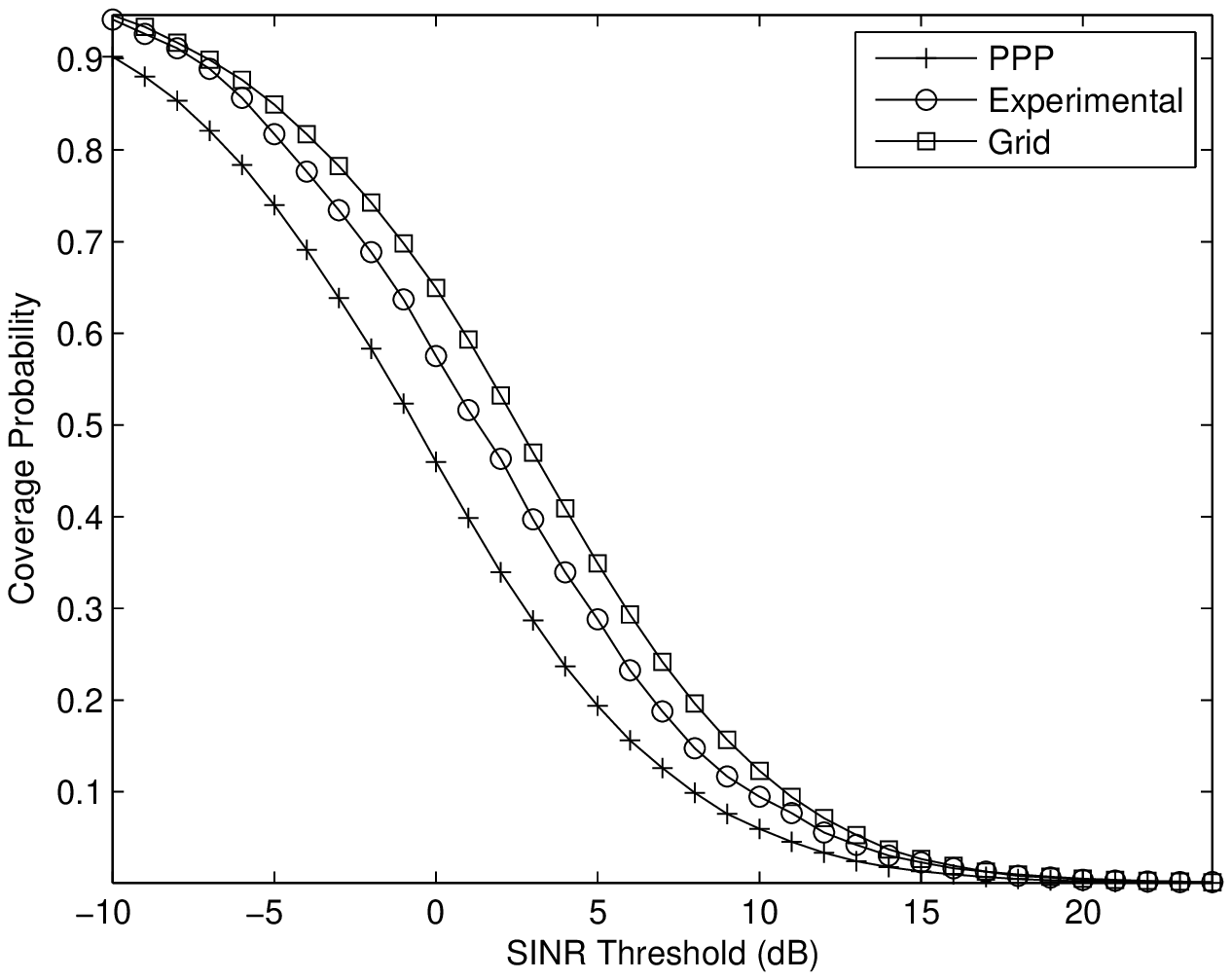}
\label{fig:expSFR}
}
\label{fig:exp}
\caption{Edge user coverage probability comparison between Poisson-model, grid model, and actual base station locations for Strict FFR (left) and SFR (right) with $\TR = 1$dB, $\Delta = 3$, no noise, and $\alpha=4$.}
\end{figure}

\begin{figure}
\centering
\includegraphics[width=5in]{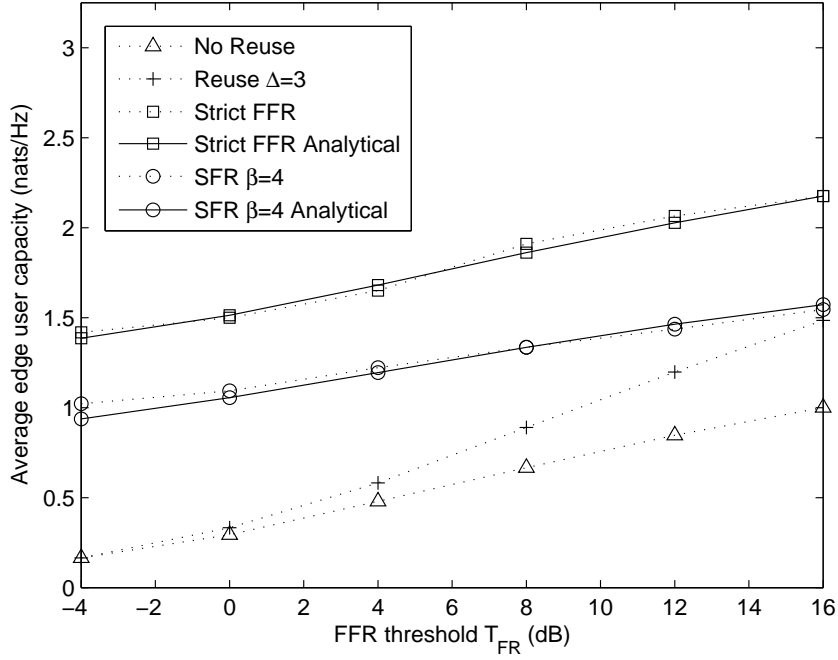}
\caption{Average edge user capacity for different reuse strategies with no noise and $\alpha=4$ as a function of a function of the $\sinr$ threshold $\TR$. The rates for Strict FFR and SFR are additionally compared with the analytical expressions derived in Section \ref{sec:rate}.}
\label{fig:Ecap}
\end{figure}

\begin{figure}
\centering
\includegraphics[width=5in]{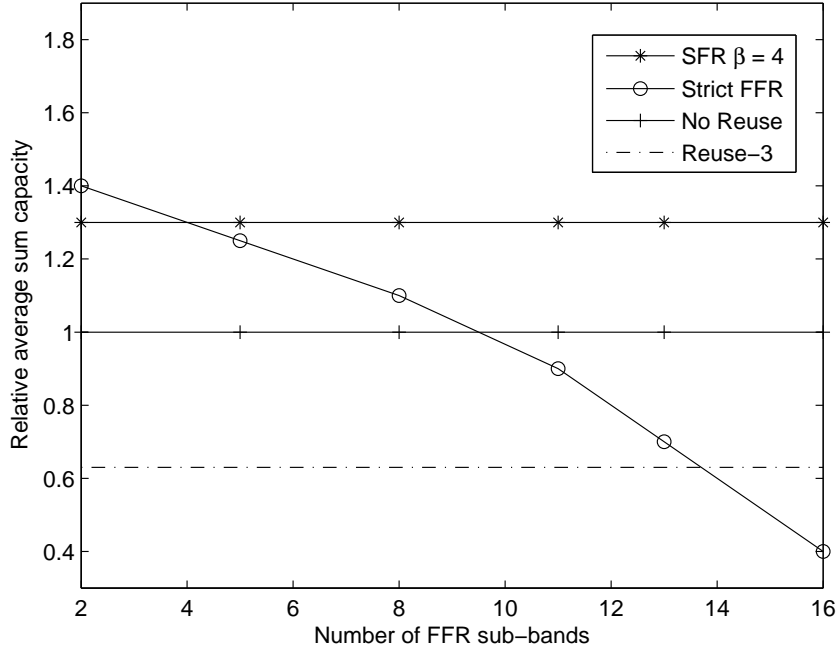}
\caption{Average sum capacity for different reuse strategies with $\TR = 3$dB, no noise, and $\alpha=4$ as a function of the number of FFR sub-bands allocated for edge users.}
\label{fig:rSumcap}
\end{figure}

\begin{figure}
\centering
\includegraphics[width=5in]{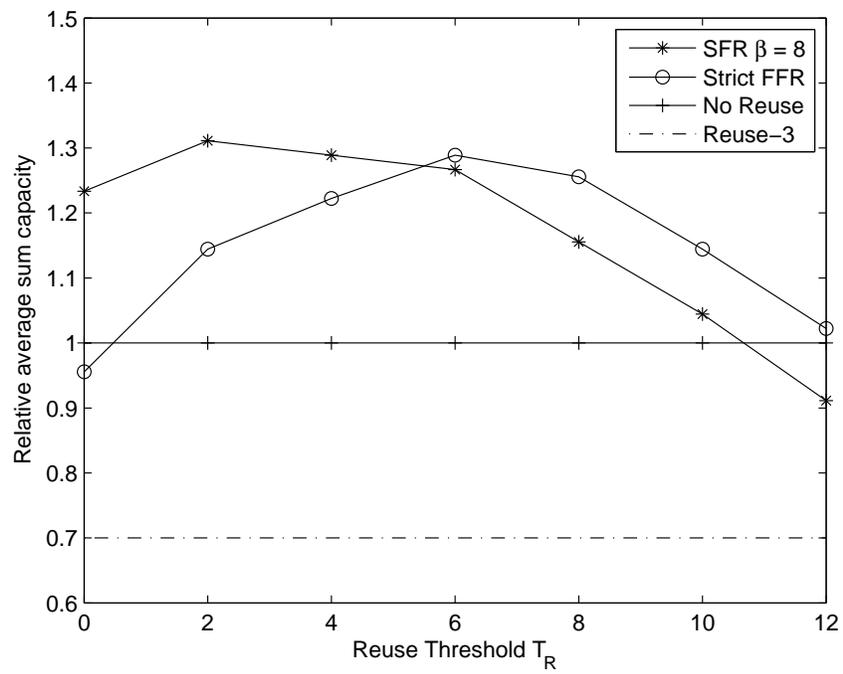}
\caption{Average sum capacity for different reuse strategies using $\sinr$-proportional sub-band allocation as a function of the $\sinr$ threshold $\TR$.}
\label{fig:tSumcap}
\end{figure}

\end{document}